\documentclass[11pt]{article}

\usepackage[utf8]{inputenc}

\usepackage{amsmath} % AMS Math Package
\usepackage{amsthm,bbm} % Theorem Formatting
\usepackage{amssymb}	% Math symbols such as \mathbb
\usepackage{graphicx} % Allows for eps images
\usepackage{multicol} % Allows for multiple columns
\usepackage[a4paper]{geometry}
\usepackage{stmaryrd}

\usepackage{xcolor}
\usepackage[backref]{hyperref}

\hypersetup{
    unicode=true,
    pdftitle={A universal adiabatic quantum query algorithm}, 
    pdfauthor={Mathieu Brandeho, J\'er\'emie Roland},
    colorlinks=true,      			% false: boxed links; true: colored links
    linkcolor=red!50!black,         % color of internal links
    citecolor=blue!50!black,        % color of links to bibliography
    urlcolor=blue!50!black          % color of external links
}

\let\baraccent=\= % rename builtin command \= to \baraccent
\newcommand{\abs}[1]{\left| #1 \right|} % for absolute value
\newcommand{\avg}[1]{\left< #1 \right>} % for average
\newcommand{\ket}[1]{\left| #1 \right>}
\newcommand{\bra}[1]{\left< #1 \right|} 
\newcommand{\braket}[2]{\left< #1 \vphantom{#2} \right| \left. #2 \vphantom{#1} \right>} 
\newcommand{\ketbra}[2]{ \vert#1 \vphantom{#2}\!\left> \!\right< \!#2 \vphantom{#1} \vert} 
\newcommand{\proj}[1]{\left| #1 \right>\!\!\left< #1 \right|}
 
\renewcommand{\=}[1]{\stackrel{#1}{=}} % for putting numbers above =
\newcommand\bbone{\ensuremath{\mathbbm{1}}}  % Identite
\newcommand\tr{\mathrm{tr}}
\newcommand{\scprod}[2]{\left\langle #1,#2\right\rangle}
\newcommand{\norm}[1]{\left\|#1\right\|}

\newcommand{\I}{\mathcal{I}}
\newcommand{\A}{\mathcal{A}}
\newcommand\eps{\varepsilon}

\newcommand{\allone}{\mathbb{J}}
\newcommand{\X}{\mathbb{X}}

\newcommand{\advstar}{\mathrm{Adv}^\star}
\newcommand{\ct}{\mathrm{ct}}
\newcommand{\dt}{\mathrm{dt}}
\newcommand{\nc}{\mathrm{nc}}

\renewcommand{\Re}{\operatorname{Re}}
\renewcommand{\H}{\mathcal{H}}

\newtheorem{thm}{Theorem}[section]
\newtheorem{lem}[thm]{Lemma}
\newtheorem{prop}[thm]{Proposition}
\newtheorem{cor}[thm]{Corollary}
\theoremstyle{definition}
\newtheorem{dfn}[thm]{Definition}
\newtheorem{cla}[thm]{Claim}
\newtheorem{claim}[thm]{Claim}
\newtheorem{fact}[thm]{Fact}
\theoremstyle{remark}

\title{A universal adiabatic quantum query algorithm}
\author{Mathieu Brandeho\footnote{\texttt{mbrandeh@ulb.ac.be}}~~and J\'er\'emie Roland\footnote{\texttt{jroland@ulb.ac.be}}\\
\small{\textit{Quantum Information and Communication, Ecole Polytechnique de Bruxelles,}}\\
\small{\textit{ Universit\'e libre de Bruxelles, 1050 Brussels, Belgium}}}
\date{}

\begin{document}

\maketitle

\begin{abstract}
Quantum query complexity is known to be characterized by the so-called quantum adversary bound. While this result has been proved in the standard discrete-time model of quantum computation, it also holds for continuous-time (or Hamiltonian-based) quantum computation, due to a known equivalence between these two query complexity models. In this work, we revisit this result by providing a direct proof in the continuous-time model. One originality of our proof is that it draws new connections between the adversary bound, a modern technique of theoretical computer science, and early theorems of quantum mechanics. Indeed, the proof of the lower bound is based on Ehrenfest's theorem, while the upper bound relies on the adiabatic theorem, as it goes by constructing a universal adiabatic quantum query algorithm. Another originality is that we use for the first time in the context of quantum computation a version of the adiabatic theorem that does not require a spectral gap.
\end{abstract}

\section{Introduction}

The quantum adversary method was originally introduced by Ambainis \cite{Amb02} for lower-bounding the \textit{quantum query complexity} $Q(f)$ of a function $f$. It is based on optimizing a matrix $\Gamma$ assigning weights to pairs of inputs. It was later shown by H\o yer et al. \cite{HLS07} that using negative weights also provides a lower bound, which is stronger for some functions. A series of works \cite{Rei09,Rei11,RS12} then led to the breakthrough result that this \textit{generalized} adversary bound, which we will simply call adversary bound from now on, actually characterizes the quantum query complexity of any function $f$ with boolean output and binary input alphabet. This is shown by constructing a tight algorithm based on the dual of the semidefinite program corresponding to the adversary bound\footnote{Note that constructing a tight algorithm for a specific problem using this method requires to find an optimal feasible point for the semidefinite program, so that this method is not necessarily constructive. The same limitation will affect the universal adiabatic algorithm in the present article.}. Finally, Lee et al.~\cite{LMRSS11} have generalized this result to the quantum query complexity of state conversion, where instead of computing a function $f(x)$, one needs to convert a quantum state $\ket{\rho_x}$ into another quantum state $\ket{\sigma_x}$.

All these results where obtained in the usual discrete-time query model, where each query corresponds to applying a unitary oracle $O_x$. In this model, an algorithm then consists in a series of input-independent unitaries $U_1,U_2,\ldots,U_T$, interleaved with oracle calls $O_x$. Another natural model is the continuous-time (or Hamiltonian-based) model where the oracle corresponds to a Hamiltonian $H_x$, and the algorithm consists in applying a possibly time-dependent, but input-independent, \textit{driver} Hamiltonian $H_D(t)$, together with the oracle Hamiltonian. The two models are related by the fact that the unitary oracle $O_x$ can be simulated by applying the Hamiltonian oracle $H_x$ for some constant amount of time. This implies that the continuous-time model is at least as powerful as the discrete-time model. In the other direction, Cleve et al.~\cite{CGM+09} have shown that the discrete-time model can simulate the continuous-time model up to at most a sublogarithmic overhead, which implies that the continuous- and discrete-time models are equivalent up to a sublogarithmic factor. Lee et al.~\cite{LMRSS11} later improved this result to a full equivalence of both models, by showing that the fractional query model, an intermediate model proved in~\cite{CGM+09} to be equivalent to the continuous-time model, is also lower bounded by the adversary bound, so that all these models are characterized by this same bound (in the case of functions, a similar result can be obtained by extending an earlier proof of Yonge-Mallo, originally considering the adversary bound with positive weights, to the case of negative weights~\cite{YM11}).

Even though these results imply that the continuous-time quantum query complexity is characterized by the adversary bound, they do not provide an explicit Hamiltonian-based query algorithm, except the one obtained from the discrete-time algorithm by replacing each unitary oracle call by the application of the Hamiltonian oracle for a constant amount of time. The resulting Hamiltonian of this algorithm then involves many discontinuities (at all times in between unitary gates), which is not very satisfying from the point of view of physics, where \textit{reasonable} Hamiltonians are smooth. However, such discontinuities are not unavoidable, as for some problems, continuous-time query algorithms based on smooth Hamiltonians are known.

The first example is unstructured search, for which Farhi and Gutmann~\cite{FG96} proposed a continuous-time analogue of Grover's algorithm based on a simple time-independent Hamiltonian (later, van Dam et al. \cite{VanDam2002}, as well as Roland and Cerf~\cite{RC02}, independently proposed an adiabatic version of this algorithm, based on a slowly varying Hamiltonian). Algorithms were also developed in the continuous-time model for various problems such as spatial search~\cite{Childs,Childs2004b,Foulger2014}, oracle identification~\cite{Mochon2007}, or element distinctness~\cite{Childs2009}. In a seminal paper, Farhi et al.~\cite{FGG07} proposed a quantum algorithm for the NAND-tree based on scattering a wave incoming on the tree, using a time-independent Hamiltonian. It is precisely this algorithm that, through successive extensions, led to the tight algorithm based on the adversary bound for any function in~\cite{Rei11}, but most of these extensions were using the discrete-time model.

In this article, we give a new continuous-time quantum query algorithm for any state conversion problem based on a slowly varying Hamiltonian, and also provide a direct proof of its optimality based on Ehrenfest’s theorem, hence proving that the quantum query complexity of any state conversion problem is characterized by the adversary bound.  The soundness of the adiabatic evolution used in our algorithm relies on a lemma from Avron and Elgart \cite{Avron1998}, which does not require the usual gap condition but only weaker spectral conditions, and was originally introduced to study atoms in quantized radiation fields. To the best of our knowledge, it is the first time that such an adiabatic theorem without a gap condition is used in the context of quantum computation.

The structure of the article is as follows. Section~\ref{sect:prelim} is devoted to preliminaries: in Subsection~\ref{ssect:prelim}, we define the necessary mathematical notions; in Subsection~\ref{sect:adia}, we recall the quantum adiabatic evolution and quantum adiabatic theorems; in Subsection~\ref{sect:query}, we recall notions of quantum query complexity; and in Section~\ref{sect:method}, the discrete-time adversary method. Original contributions start in Section~\ref{sect:ct-adversary}, where we give a direct proof that the adversary bound remains a lower bound for continuous-time quantum query complexity (Theorem~\ref{thm:ct-adversary}). Finally, in Section~\ref{sect:adiaconvert}, we present our adiabatic quantum query algorithm \textbf{AdiaConvert}, and show that it is optimal, implying the characterization of the bounded-error quantum query complexity (Theorem~\ref{thm:adv}).

\section{Preliminaries}\label{sect:prelim}

\subsection{Definitions}\label{ssect:prelim}
Throughout this article, $\Sigma$ is a finite set representing the input alphabet, $\X\subset \Sigma^n$ is a subset of strings of length $n$, and $x\in\X$ denotes a possible input string.
\begin{dfn}(Matrix norms and inner product) Let $A$ and $B$ be $n$-by-$n$ matrices
\begin{itemize}
\item Inner product:  $\left\langle A,B\right\rangle= \tr(A^* B)$, where $A^*$ is the adjoint matrix of $A$,
\item Hadamard product: $(A\circ B)_{ij}= A_{ij}\,.\, B_{ij}$,
\item Operator norm:  $\|A\|= \max_{\ket{v}}\frac{\|A\ket{v}\|}{\|\ket{v}\|}= \max_{\ket{u},\ket{v}}\frac{\bra{u}A\ket{v}}{\|\ket{u}\| .\|\ket{v}\|}$,
\item Trace norm: $\|A\|_\tr = \max_B \frac{\left\langle A,B\right\rangle}{\|B\|}$.
\end{itemize}
\end{dfn}

These definitions imply the following properties
\begin{lem}
 For any $n$-by-$n$ matrices $A,B,C$, we have
\begin{itemize}
 \item $\left\langle A\circ C,B\right\rangle=\left\langle A,B\circ C^*\right\rangle$
 \item $\left\langle A,B\right\rangle \leq \|A\|_\tr\cdot \|B\|$
\end{itemize}
\end{lem}

In this context, the following matrix norm will be useful:
\begin{dfn}[$\gamma_2$ norm]\label{def:gamma2}
Let $\mathcal{D}$ be a finite set, $A$ a $\vert \mathcal{D}\vert$-square matrix. The norm $\gamma_2(A)$ is defined as
\[
\begin{split}
\gamma_2(A)&= \min_{\substack{m\in \mathbb{N}  \\ \ket{u_{x}},\ket{v_{y}} \in \mathbb{C}^m}}\Bigg\{
\max_{x\in \mathcal{D} } \,\max \Big\{  \|\ket{u_{x }} \|^2 ,    \|\ket{v_{y }} \|^2\Big\}
\Bigg\vert \forall\, x,y\in\mathcal{D}, \, A_{x,y}=\braket{u_x}{v_y} \Bigg\},\\
& =\max_{\substack{\ket{u},\ket{v}  \\ \|\ket{u}\|=\|\ket{v}\|=1}} \|A\circ\ketbra{u}{v}\|_\tr.
\end{split}\]
\end{dfn}

In particular, it is shown in~\cite{LMRSS11} that the dual of the Adversary bound can be seen as a variation of the $\gamma_2$ norm dubbed the filtered $\gamma_2$ norm.
\begin{dfn}[Filtered $\gamma_2$ norm]\label{def:filter}
Let $\mathcal{D}_1$ and $\mathcal{D}_2$ be two finite sets, $A,\,Z_1,\dots,Z_n$  matrices  with $\vert \mathcal{D}_1\vert$ rows and  $\vert \mathcal{D}_2\vert$ columns, and  $Z=\{Z_1,\dots,Z_n\}$. The norm $\gamma_2(A\vert Z)$ is defined as
\[
\begin{split}
\gamma_2(A\vert Z)=& \min_{\substack{m\in \mathbb{N}\\ \ket{u_{x,j}},\ket{v_{y,j}} \in \mathbb{C}^m}}
\max \Big\{ \max_{x\in \mathcal{D}_1}\sum_j \|\ket{u_{x,j}} \|^2 , \max_{y\in \mathcal{D}_2}\sum_j \|\ket{v_{y,j}} \|^2\Big\}\\
&\mathrm{subject\,\, to}\quad\forall (x,y)\in \mathcal{D}_1\times \mathcal{D}_2, \quad A_{x,y}=\sum_j (Z_j)_{x,y}  \braket{u_{x,j}}{v_{y,j}},\\
=&\qquad\max_\Gamma  \quad \| \Gamma\circ A \| \qquad\mathrm{subject\,\, to}\quad \forall j \hspace{0.2cm} \|\Gamma \circ Z_j \| \leq 1.
\end{split}
\]
\end{dfn}

\begin{cla}\cite{LR13} \label{cla:LR} For any matrices $A$, $B$ where $A\circ B$ is defined, $\|A\circ B\|\leq \gamma_2(A).\|B\|$.
\end{cla}

The Hadamard product fidelity is introduced in \cite{LR13} to characterize the output condition of quantum query problems. Whereas the usual fidelity compares density matrices, the Hadamard product fidelity compares Gram matrices (note that if $\rho$ is a Gram matrix and $\ket{u}$ is a normalized state, then $\rho\circ\ketbra{u}{u}$ is a density matrix).
\begin{dfn}[Hadamard product fidelity]
The Hadamard product fidelity between two Gram matrices $\rho$ and $\sigma$ is defined as
\[\mathcal{F}_H(\rho,\sigma)= \min_{\ket{u}:\|\ket{u}\|=1}\mathcal{F}(\rho\circ \ketbra{u}{u},\sigma\circ \ketbra{u}{u}),\]
where $\mathcal{F}(\rho',\sigma')$ is the fidelity between two density matrices $\rho'$ and $\sigma'$, defined as $\mathcal{F}(\rho',\sigma')=\tr \sqrt{\sqrt{\rho'}~\sigma'\sqrt{\rho'}}$.
\end{dfn}

We similary define the Hadamard product distance from the trace distance.
\begin{dfn}[Hadamard product distance]\label{def:trace}
The Hadamard product distance between two Gram matrices $\rho$ and $\sigma$ is defined as
\[\mathcal{D}_H(\rho,\sigma)= \max_{\ket{u}:\|\ket{u}\|=1}\mathcal{D}(\rho\circ \ketbra{u}{u},\sigma\circ \ketbra{u}{u}),\]
where $\mathcal{D}(\rho',\sigma')$ is the trace distance between two density matrices $\rho'$ and $\sigma'$, defined as $\mathcal{D}(\rho',\sigma')=\frac{1}{2} \|\rho'-\sigma'\|_\tr $.
\end{dfn}

\begin{thm} \cite{Fuchs1998} For any density matrices $\rho$, $\sigma$, we have  $1-\mathcal{D}(\rho,\sigma)\leq \mathcal{F}(\rho,\sigma)\leq \sqrt{1-\mathcal{D}^2(\rho,\sigma)}$.
\end{thm}
%This relation can be easily extended to $\mathcal{D}_H$ and $\mathcal{F}_H$.

\begin{cor}\label{cor:relation}For any Gram matrices $\rho$, $\sigma$, we have $1-\mathcal{D}_H(\rho,\sigma)\leq \mathcal{F}_H(\rho,\sigma)\leq \sqrt{1-\mathcal{D}_H^2(\rho,\sigma)}$.
\end{cor}

\begin{dfn}[Distance between quantum states]
We say that two normalized quantum states $\ket{\phi},\ket{\psi}\in\H$ are $\eps$-distant if
$
 \norm{\ket{\phi}-\ket{\psi}}\leq\eps.
$
\end{dfn}

\subsection{Quantum query complexity}\label{sect:query}

In classical computation, a query algorithm computes a function $f: \X\subset \Sigma^n \rightarrow B$ where the input $x\in \X$ can only be accessed through queries to an oracle that, on input $j\in [n]$, outputs $x_j\in\Sigma$. A query algorithm can be seen as a decision tree \cite{BdW02} where each vertex represents a decision taken after one query. The depth of the tree then corresponds to the number of queries used by this algorithm to compute $f$ in the worst case. The \textit{query complexity of $f$} is the minimum depth of all decision trees computing $f$ exactly. 

In quantum computation, query complexity can be generalized to state conversion problems, where one should convert a quantum state $\ket{\rho_x}$ into another state $\ket{\sigma_x}$, each depending on the input $x$, which can once again only be accessed via an oracle. The evaluation of a function $f$ is the particular case where initial states are independent of $x$, and final states are orthonormal for $x,y$ such that $f(x)\neq f(y)$. For any set of quantum states $\{\ket{\rho_x}\}_x$, it is enough to consider the Gram matrix $\rho_{x,y}=\braket{\rho_x}{\rho_y}$, because if $\braket{\rho_x}{\rho_y}=\braket{\rho'_x}{\rho'_y}$ for all $x,y$, then there exists a unitary transformation $U$ independent of $x$ such that $\ket{\rho_x}=U\ket{\rho'_x}$ for all $x$. This implies that a query algorithm for the set of states $\{\ket{\rho}\}_x$ can be converted into a query algorithm for the set of states $\{\ket{\rho'}\}_x$ without additive cost, and \textit{vice versa}. We will therefore denote by a pair of Gram matrices $(\rho,\sigma)$ the problem of converting a set of states  $\{\ket{\rho_x}\}_x$ into another set of states  $\{\ket{\sigma_x}\}_x$.

In the discrete-time model of quantum query complexity, we can consider without loss of generality an oracle $O_x$ acting on an $n$-dimensional input register and a $(|\Sigma|+1)$-dimensional output register as
\begin{align}\label{eq:uniQ}
 O_x:
\left\{
\begin{array}{ll}
\ket{j}\ket{\bar{0}}\mapsto\ket{j}\ket{x_j} &\forall j\in [n]\\
 \ket{j}\ket{x_j}\mapsto\ket{j}\ket{\bar{0}} &\forall j\in [n]\\
 \ket{j}\ket{y}\mapsto\ket{j}\ket{y} &\forall j\in [n],y\in\Sigma\setminus\{x_j\}
\end{array}
\right.
\end{align}
where $\bar{0}$ is an additional output alphabet symbol, that can be seen as a blank symbol. A query algorithm in this model is then given by a succession of input-independent unitaries $U_t$ interleaved with oracle calls $O_x$. The discrete-time quantum query complexity $Q^\dt_0(\rho,\sigma)$ is the minimum number of oracle calls of any such algorithm converting $\rho$ to $\sigma$ exactly. (Note that there exist alternative definitions for the oracle $O_x$, but they only affect the definition of $Q^\dt_0(\rho,\sigma)$ by at most a constant factor.)

In the continuous-time model, the oracle is a Hamiltonian $H_\mathcal{Q}(x)$ of the general form
\begin{equation}\label{hamiQ}
H_\mathcal{Q}(x)= \sum_{j=1}^n{\proj{j}\otimes h(x_j)},
\end{equation}
where each $\{h(y)\}_{y\in\Sigma}$ is hermitian and satisfies $\|h(y)\|\leq 1$. In particular, the choice $h(y)=\proj{y^-}$, where
\begin{equation}\label{defPM}
 \ket{y^\pm}=\frac{1}{\sqrt{2}}(\ket{\bar{0}}\pm\ket{y}),
\end{equation}
can be considered as the Hamiltonian analogue of the unitary oracle $O_x$ in equation~(\ref{eq:uniQ}), since it is easy to check that $O_x=e^{-iH_\mathcal{Q}(x)\Delta T}$ for $\Delta T=\pi$. A query algorithm in this model then corresponds to applying a Hamiltonian $H_x(t)$ of the form
\begin{equation}\label{hami}
H_x(t) = H_\mathcal{D}(t)+ \alpha(t)H_\mathcal{Q}(x) %\qquad \mathrm{with}\qquad \|H_\mathcal{Q}(x)\|\leq 1,
\end{equation}
where $H_\mathcal{D}(t)$ is the driver Hamiltonian independent of the input $x$, and $\abs{\alpha(t)}\leq 1$  for all $t\in[0,T]$.  The continuous-time quantum query complexity $Q^\ct_0(\rho,\sigma)$ is the minimum computing time $T$ of any such algorithm converting $\rho$ to $\sigma$ exactly.

For scenarios where we accept errors, we must distinguish two cases : \textit{coherent} and \textit{non-coherent} quantum state conversion. Concretely, a computation will typically use some extra workspace and may therefore generate a state  $\ket{\sigma_x, J_x}$, where $\ket{J_x}$ is the final state of the workspace. This might not be desirable if the state generation is used as a subroutine in a larger quantum algorithm, where we would like to use interferences between the states $\ket{\sigma_x}$ for different $x$'s. In that case, we would like to be able to reset the state $\ket{J_x}$ to a default state, so that it does not affect interferences.

We therefore define the following output conditions (both for the discrete- and continuous-time models)
\begin{dfn}[Output condition]\label{def:output}
A quantum query algorithm acting as unitary $\mathcal{U}_x$ for input $x$ converts  $\rho$ to $\sigma$  with error at most $\varepsilon$ if
\begin{itemize}
\item (coherent case) $\forall x\in \X, \, \Re( \braket{\sigma_x,0}{\mathcal{U}_x\vert\rho_x,0})\geq \sqrt{1-\varepsilon}$,
\item (non-coherent case) $\forall x\in \X, \, \exists \ket{J_x}, \, \Re(\braket{\sigma_x,J_x}{\mathcal{U}_x\vert\rho_x,0})\geq \sqrt{1-\varepsilon}$.
\end{itemize}
\end{dfn}

Note that a sufficient condition for $\Re( \braket{\phi}{\psi})\geq \sqrt{1-\varepsilon}$ is that these states are $\sqrt{\eps}$-distant. Moreover, the output condition for the coherent case has been shown~\cite{LR13} to be equivalent to $\mathcal{F}_H(\sigma,\sigma')\geq \sqrt{1-\epsilon}$, where $\sigma'$ is the Gram matrix of the output states $\ket{\sigma'_x}=\mathcal{U}_x\ket{\rho_x,0}$. Similarly, in the non-coherent case the output conditions can be rewritten as $\mathcal{F}_H(\sigma\circ J,\sigma')\geq \sqrt{1-\epsilon}$, where $J$ is any Gram matrix of unit vectors (corresponding to any set of states $\ket{J_x}$). This implies that bounded-error and zero-error quantum query complexities are related as follows.
\begin{lem}[\cite{LR13}]\label{lem:bounded-error}
 For any $|\X|$-by-$|\X|$ Gram matrices $\rho,\sigma$, we have
 \begin{align}
Q^\bullet_\eps(\rho,\sigma)&=\min_{\sigma'} \left\{Q^\bullet_0(\rho,\sigma'):\mathcal{F}_H(\sigma,\sigma')\geq \sqrt{1-\epsilon}\right\}\\
Q^{\nc,\bullet}_\eps(\rho,\sigma)&=\min_{\sigma'} \left\{Q^{\bullet}_0(\rho,\sigma'):\mathcal{F}_H(\sigma\circ J,\sigma')\geq \sqrt{1-\epsilon}, J\circ \bbone =\bbone\right\}
\end{align}
where the superscript $\nc$ denotes the non-coherent query complexity (otherwise we consider the coherent case by default), and the superscript $\bullet$ is either $\dt$ or $\ct$.
\end{lem}

Computing a function $f$ is equivalent to generating the Gram matrix $F_{x,y}=\delta_{f(x),f(y)}$ from the all-1 Gram matrix $\allone_{x,y}=1$. In that case, it is not necessary to generate the state coherently, but one can convert a non-coherent algorithm into a coherent algorithm, so that we can consider the coherent case without loss of generality.
\begin{lem}[\cite{LR13}]\label{lem:coherent-function}
 For any function $f$ and associated Gram matrix $F_{x,y}=\delta_{f(x),f(y)}$, we have $Q^\bullet_\eps(f)=Q^{\nc,\bullet}_\eps(\allone,F)$ and
\begin{align*}
 Q^{\nc,\bullet}_\eps(\allone,F)\leq Q^{\bullet}_\eps(\allone,F)\leq 2 Q^{\nc,\bullet}_{1-\sqrt{1-\eps}}(\allone,F).
\end{align*}

\end{lem}

\subsection{Adversary methods}
\label{sect:method}

The quantum adversary method is one of main methods to prove lower bounds on quantum query complexity (the other main method is the polynomial method~\cite{BBC+01}). Its basic principle is rather simple: it consists in defining a so-called progress function $W$ whose value is high at the beginning of the algorithm and should be low at the end of the algorithm if it is successful. By bounding the change in the progress function for each oracle call, one then bounds the minimum number of oracle calls necessary for success.

More precisely, let $\ket{\phi_x(t)}$ be the state of the algorithm on input $x$ after $t$ queries, and $\Phi_t$ be the Gram matrix of those states. We define a progress function
\begin{align*}
 W(\Phi_t)=\scprod{\Gamma\circ vv^*}{\Phi_t},
\end{align*}
where $\Gamma$ is a $\vert \X\vert$-by-$\vert \X\vert$ hermitian matrix, called the \textit{adversary matrix}, and $v$ a unit vector. We also define the matrices $\Delta_j$ with entries $(\Delta_j)_{x,y}=1-\delta_{x_j,y_j}$. The adversary method relies on the fact that if $\Gamma$ is chosen so that it satisfies $\| \Gamma \circ \Delta_j\| \leq 1$ for all $j\in [n]$, then the progress function can only increase by one after each query (see e.g.~\cite{HLS07}), that is, $\vert W(\Phi_{t+1})-W(\Phi_t) \vert\leq 1$. The difference of the values of the progress function between $\Phi_0=\rho$ and $\Phi_T=\sigma$ is then given by
\begin{align*}
 W(\Phi_0)-W(\Phi_T)=\scprod{\Gamma\circ vv^*}{\rho-\sigma}=\scprod{\Gamma\circ(\rho-\sigma)}{vv^*}\leq T
\end{align*}
By optimizing over $\Gamma$ and $v$, we obtain the adversary bound
\begin{dfn}\cite{LMRSS11,LR13}(Adversary bound) \label{def:adv}
\[
\begin{split}
\advstar(\rho,\sigma) &= \max_{\Gamma} \left\|\Gamma \circ(\rho-\sigma)\right\|\hspace{1.5cm} \text{subject to} \hspace{1cm}   \forall j\in[n], \hspace{0.3cm}\| \Gamma \circ \Delta_j\| \leq 1,\\
&= \gamma_2(\rho-\sigma|\Delta) \qquad\qquad\qquad\text{where} \qquad\Delta= \{\Delta_1,\dots,\Delta_n\}.
\end{split} \]
\end{dfn}

As shown in~\cite{LMRSS11}, $\advstar$ defines a distance between Gram matrices, sometimes called \textit{the query distance}. The following simple proposition, comparing the query distance to the Hadamard product distance $\mathcal{D}_H$, will be used in the proof of Theorem~\ref{thm:adv}.
\begin{prop}\label{lem:Dist} For any Gram matrices $\rho$, $\sigma$  of size $n$, $\mathcal{D}_H(\rho,\sigma)\leq \advstar(\rho,\sigma)$. 
\end{prop} 
\begin{proof}
Since the trace distance may be written as $\mathcal{D}(\rho' ,\sigma' )= \max_{P:\|P\|\leq 1}\frac{1}{2}\left\langle P,(\rho'-\sigma')\right\rangle $, we can reformulate the Hadamard product distance in Definition \ref{def:trace} as
\begin{align*}
\mathcal{D}_H(\rho ,\sigma)= \max_{\substack{P :\|P\|\leq 1/2\\ \ket{u} :\|\ket{u}\|=1}}\left\langle P, (\rho-\sigma)\circ \ketbra{u}{u}\right\rangle=\max_{\substack{P :\|P\|\leq 1/2 }}\|P\circ (\rho-\sigma)\|.
\end{align*}

We observe that this form is similar to $\advstar$ in Definition \ref{def:adv}, except for the constraints on $P$ and $\Gamma$. We conclude the proof by showing that the constraint on $P$ is stronger, that is, if $\|P\|\leq 1/2$ then $\|P\circ \Delta_i\|\leq 1$ for all $i\in [n]$.

Let $\allone$ be the all-one matrix, and $i\in[n]$. We have
\begin{align*}
\|P\circ \Delta_i\|\leq \|P\| + \|P\circ (\allone-\Delta_i)\|\leq \Big(1+\gamma_2(\allone-\Delta_i)\Big)\|P\|,
\end{align*}
where the inequalities follows from the triangle inequality and Claim \ref{cla:LR}, respectively. We finally bound $\gamma_2(\allone-\Delta_i)$ using the minimization form in Definition~\ref{def:gamma2} and an appropriate choice for $\{\ket{u_x},\ket{v_x}\}_x$. Choosing $\ket{u_x}=\ket{v_x}=\ket{x_i}$, we have $\braket{u_x}{v_y}=(\allone-\Delta_i)_{x,y}=\delta_{x_i,y_i}$, so that $\gamma_2(\allone-\Delta_i)\leq 1$.
\end{proof}

\subsection{Adiabatic quantum computation}\label{sect:adia}
Adiabatic quantum computation is a quantum computational model originally proposed by Farhi et al. \cite{FGGS00} for solving instances of the satisfiability problem.  This model is based on the quantum adiabatic theorem introduced by Born and Fock \cite{BF28} and describing a physical system evolving under a slowly varying Hamiltonian:

\begin{center}
\textit{A quantum system with a time-dependent Hamiltonian remains in its instantaneous eigenstate if the Hamiltonian variation is slow enough and there is a large gap between its eigenvalue and the rest of the spectrum of the Hamiltonian.}
\end{center}
It was later proved that the adiabatic model is equivalent to standard quantum computation \cite{AVK+04}. This statement, as well as the correctness of most adiabatic algorithms, rely on the existence of a spectral gap.

In order to formally describe adiabatic quantum computation, let us first define the notion of adiabatic process.
\begin{dfn}\label{def:adia}
An \textbf{adiabatic process} on the Hilbert space $\mathcal{H}$ is defined by a triplet $\{H(s),P(s),\tau\}$ with $s\in [0,1]$ where\vspace{0.2cm}
\begin{enumerate}
\item $H(s)$ is a twice differentiable map from $[0,1]$ to the space of bounded linear self-adjoint operators $B(\mathcal{H})$, equipped with the graph norm of $H(0)$.
\item $P(s)$ are a family of orthogonal rank-one projections onto an eigenvector of $H(s)$ with continuous eigenvalue $\lambda(s)$,
\item $\tau\in\mathbb{R}^+$ is the time scale, which defines the time as $t(s)=s\tau$.
\end{enumerate}
\end{dfn}

For such an adiabatic process, we can define the unitary operator $U_A(s)$ corresponding to an idealized evolution, which maps the eigenvector in the range of $P(0)$ to the eigenvector in the range of $P(s)$, that is, $U_A(s)P(0)U_A^*(s)=P(s)$. Furthermore, the physical evolution, represented by unitary operator $U_\tau(s)$, can be obtained from the Schr\"odinger equation
\begin{equation}\label{tau}
i\partial_s U_\tau(s)=\tau H(s)U_\tau(s).
\end{equation}
Let us note that the analytical conditions given in Definition \ref{def:adia}
ensure the existence and uniqueness of the solution $U_\tau(s)$ of this equation with initial condition $U_\tau(0)=\bbone$~\cite{reed1975methods}. 

The quantum adiabatic theorem can be summarized by the following statement
\begin{align*}
 \lim_{\tau \rightarrow \infty}U_\tau(s)P(0)=U_A(s)P(0)=P(s)U_A(s).
\end{align*}

Thus $U_\tau(s) P(0)$ converge to $U_A(s)P(0)$ for large $\tau$, and the norm of their difference defines the error of the adiabatic process. 
\begin{dfn}
The \textbf{error} $\eps_{AP}(s)$ of an adiabatic process $\{H(s),P(s),\tau\}$ is defined as
\[
\eps_{AP}(s)=\big\|\big[U_\tau(s)-U_A(s)\big]P(0)\big\|, \qquad \text{with} \quad \eps_{AP}=\eps_{AP}(1).\]
\end{dfn}
This definition implies that at the end of the adiabatic evoltion, the physical state will be $\eps_{AP}$-distant from the ideal state.

How slow should the process be, or, equivalently, how large should $\tau$ be, to ensure a small enough adiabatic error? The \textit{folk adiabatic condition} requires the following bound:
\begin{equation}\label{folk}
\tau >\!\!> \int_0^1{\frac{\|\partial_s H_\tau (s)\|}{g(s)^2}ds},
\end{equation}
where the gap $g(s)$ represents the minimal distance between the eigenvalue $\lambda(s)$ and the rest of spectrum of $H(s)$. However this \textit{folk adiabatic condition} is not always sufficient, but rigorous conditions have been given e.g. by Jansen et al. \cite{JRS07}. Indeed, they proved the following statement (where we introduce the notation $\dot{A}(s)=\partial_s A(s)$).
\begin{thm}\cite{JRS07}
Let $\{H(s),P(s),\tau\}$ be an adiabatic process with a gap $g=\min_{s\in[0,1]}g(s)$,     $\dot{H}$, $\ddot{H}$ are bounded operators, and $\eps>0$, if
\begin{align*}
\tau \geq  \frac{1}{\eps}\Big[\frac{\|\dot{H}(0)\|+\|\dot{H}(1)\|}{g^2}+\max_{s\in [0,1]}\frac{\|\ddot{H}(s)\|^2}{g^2}+7\frac{\|\dot{H}(s)\|^2}{g^3}\Big]  , 
\end{align*}
then $\eps_{AP}\leq\eps$.
\end{thm}  
% This quantum adiabatic theorem was later used by Aharonov et al. to prove the polynomial equivalence between the adiabatic and quantum circuit models~\cite{AVK+07}.

 The adiabatic process used in our algorithm introduced in Section~\ref{sect:adiaconvert} does not necessarily exhibit a gap, and for this reason we use another lemma from Avron and Elgart \cite{Avron1998}.
\begin{lem}\cite{Avron1998}\label{adiBound}
Let $\{H(s),P(s),\tau\}$ be an adiabatic process and $\eps>0$. Suppose that the commutator equation 
\begin{equation}\label{eq:comm}
\dot{P}(s)P(s)=[H(s),X(s)]
\end{equation}
accepts as solution operator $X(s)$ such that both $X(s)$ and $\dot{X}(s)$ are bounded. If
\begin{align*}
 \tau \geq \max_{s\in [0,1]}\frac{1}{\eps}\Big[2\|X(s)\|+\|\dot{X}(s)P(s)\|\Big] ,
\end{align*}
then $\eps_{AP}\leq\eps$.
\end{lem}
This version of the lemma is actually a special case of the statement proved by Avron and Elgart, adapted to the case of continuous-time quantum computation. For completeness we reproduce a self-contained proof of this version of the lemma in Appendix~\ref{sect:bound}.

\section{Adversary lower bound in the continuous-time model}\label{sect:ct-adversary}
In this section we give a direct proof that the adversary method $\advstar(\rho,\sigma)$ is a lower-bound for the zero-error quantum query complexity in the continuous-time model.
\begin{thm}\label{thm:ct-adversary}
 For any $|\X|$-by-$|\X|$ Gram matrices $\rho,\sigma$, we have
\begin{align*}
 Q_0^\ct(\rho,\sigma)&\geq\frac{1}{2}\advstar(\rho,\sigma),\\
 Q_\eps^\ct(\rho,\sigma)&\geq\frac{1}{2}\min_{\sigma':\mathcal{F}_H(\sigma,\sigma')\geq \sqrt{1-\epsilon}}\advstar(\rho,\sigma').
\end{align*}
\end{thm}

\begin{proof}
 Let $\ket{\phi_x(t)}$ be the state of the algorithm on input $x$ at time $t\in[0,T]$, and $\Phi_t$ be the Gram matrix of those states. Let $\Gamma$ be a $\vert \X\vert$-by-$\vert \X\vert$ hermitian matrix and $\ket{v}$ be a $\vert \X\vert$-dimensional unit vector. We consider the following superposition of states:
\begin{align*}
\ket{\Phi_t}= \sum_{x}v_x\ket{x}_\I\ket{\phi_x(t)}_\A \qquad \text{with} \qquad \tr_\A\ket{\Phi_t}\!\!\bra{\Phi_t}=\Phi_t\circ \ketbra{v}{v},
 \end{align*}
where $\A$ is the actual register of the algorithm, while $\I$ is a (virtual) input register that is introduced for the sake of analysis.

Since each state $\ket{\phi_x(t)}$ evolves under the influence of a Hamiltonian $H_x(t)$ as in Equation~(\ref{hami}), the state $\ket{\Phi_t}$ evolves under the influence of a global Hamiltonian
\begin{equation}%\label{hami}
H(t) = \sum_x \proj{x}\otimes H_x(t).
\end{equation}
Similarly to Subsection~\ref{sect:method}, we consider a progress function
\begin{align*}
 W(\Phi_t)&=\scprod{\Gamma\circ \proj{v}}{\Phi_t}\\
&=\tr_\I\left[\Gamma(\Phi_t\circ \proj{v})\right]\\
&=\bra{\Phi_t}\Gamma\otimes\bbone_\A\ket{\Phi_t}\\
&\equiv\avg{\Gamma}_t
\end{align*}
where we use the usual notation $\avg{\Gamma}_t$ for the expectation value of observable $\Gamma$ when measuring state $\ket{\Phi_t}$.
From Ehrenfest's theorem~\cite{Ehr27}, this expectation value evolves as
\begin{align*}
 \frac{d\avg{\Gamma}_t}{dt}=-i\avg{[\Gamma,H(t)]}_t+\avg{\frac{\partial\Gamma}{\partial t}}_t,
\end{align*}
where the second term is zero since $\Gamma$ is time-independent. Therefore, we have
\begin{align*}
 \frac{dW(\Phi_t)}{dt}&=-i\bra{\Phi_t}[\Gamma,H(t)]\ket{\Phi_t}\\
&=-i\sum_{x,y}v_xv_y^*\Gamma_{yx}\bra{\phi_y(t)}H_x(t)-H_y(t)\ket{\phi_x(t)}\\
&=-i\alpha(t) \sum_{x,y}v_xv_y^*\Gamma_{yx}\sum_{j:x_j\neq y_j}\bra{\phi_y(t)} \proj{j}\otimes[h(x_j)-h(y_j)] \ket{\phi_x(t)}\\
&=-i\alpha(t)  \sum_{ j}\sum_{x,y} (1-\delta_{x_jy_j}) v_xv_y^*\Gamma_{yx}[\Phi_t^j]_{yx}\\
&=-i\alpha(t) \sum_j \scprod{\Gamma\circ\Delta_j}{\Phi_t^j\circ \proj{v}},
\end{align*}
where we have defined the matrices $[\Phi_t^j]_{yx}= \braket{ \phi_y(t)}{ \proj{j}\otimes[h(x_j)-h(y_j) ]\vert\phi_x(t)}$.
Using the properties of the inner product and the fact that $\abs{\alpha(t)}\leq 1$, we may  bound the variation of the progress function as
\begin{align*}
 \left|\frac{dW(\Phi_t)}{dt}\right|&\leq\left|\sum_j \scprod{\Gamma\circ\Delta_j}{\Phi_t^j\circ \ketbra{v}{v}}\right|\\
&\leq\sum_j \|\Gamma \circ\Delta_j \| . \|\Phi_t^j\circ \ketbra{v}{v}\|_\tr,\\
&\leq\sum_j \|\Gamma \circ\Delta_j \| . \gamma_2(\Phi_t^j),\\
&\leq\max_j\|\Gamma \circ\Delta_j \|\cdot\Big[ \sum_j \gamma_2(\Phi_t^j)\Big].
\end{align*}
We now show that $\sum_j \gamma_2(\Phi_t^j)\leq 2$. First, as $\{\proj{j}\}_{j\in[n]}$ is a set of orthogonal projectors defined from the orthogonal basis $\{\ket{j}\}_{j\in[n]}$, we have $\sum_j \gamma_2(\Phi_t^j)= \gamma_2(\sum_j\Phi_t^j)$.

Using the minimization form in Definition~\ref{def:gamma2}, we show that there exist $\{\ket{u_x},\ket{v_x}\}_x$ such that $\sum_j \big[\Phi_t^j\big]_{yx}=\braket{u_y}{v_x}$ and $\max_x \big\{\max \{\|\ket{v_x}\|^2,\|\ket{u_x}\|^2\}\big\}\leq 2$. Indeed, let
\begin{align*}
 \ket{u_x}=-H_\mathcal{Q}(x)\ket{\phi_x(t)}\ket{0}+\ket{\phi_x(t)}\ket{1}, \qquad 
\ket{v_x}=\ket{\phi_x(t)}\ket{0}+H_\mathcal{Q}(x)\ket{\phi_x(t)}\ket{1}.
\end{align*}
Then, we have $\braket{u_y}{v_x}=\sum_j[\Phi_t^j]_{yx}$, and the upper-bound  on the norms of these vectors follows from the conditions $\|h(y)\|\leq 1$ for all $y$, which imply $\norm{H_\mathcal{Q}(x)}\leq 1$ for all $x$.

Since $\sum_j \gamma_2(\Phi_t^j)\leq 2$, the last bound then reduces to
\begin{align*}
 \left|\frac{dW(\Phi_t)}{dt}\right|&\leq2\max_j\|\Gamma \circ\Delta_j\|.
\end{align*}
Moreover, for a zero-error algorithm, we also have
\begin{align*}
\label{derivADV}
\big\vert\left\langle \Gamma\circ(\sigma-\rho),vv^*\right\rangle\big\vert
&=\big\vert W(\Phi_T)-W(\Phi_0)\big\vert\\
&= \left|\int_0^T \frac{dW(\Phi_t)}{dt}\right|\\
&\leq T \max_{t\in[0,T]}\left|\frac{dW(\Phi_t)}{dt}\right|\\
&\leq 2T \max_j\|\Gamma \circ\Delta_j\|.
\end{align*}
By optimizing over $\Gamma$ and $\ket{v}$, we obtain the zero-error adversary bound $T\geq\frac{1}{2}\advstar(\rho,\sigma)$, which proves the first part of the theorem. The second part then directly follows from Lemma~\ref{lem:bounded-error}.
\end{proof}

\section{Adiabatic quantum query algorithm}\label{sect:adiaconvert}

In this section, we build an adiabatic quantum query algorithm \textbf{AdiaConvert}($\rho,\sigma,\varepsilon$), for solving the  quantum state conversion problem $(\rho,\sigma)$, with an error $\eps$ and a running time $\tau=O(\advstar(\rho,\sigma)/\eps)$.
Together with Theorem \ref{thm:ct-adversary}, this implies that the adversary method characterizes the quantum query complexity in the time-continuous model for bounded error.
\begin{thm}\label{thm:adv}
 For any $|\X|$-by-$|\X|$ Gram matrices $\rho,\sigma$, we have
\begin{align*}
 Q_\eps^\ct(\rho,\sigma)&= O\Big(\frac{\advstar(\rho,\sigma)}{\eps}\Big).
\end{align*}
\end{thm}

\paragraph*{Description of AdiaConvert}
The algorithm acts on a Hilbert space $\mathcal{H}=\mathcal{H}_\mathcal{O}\oplus\mathcal{H}_\mathcal{Q}\otimes\mathcal{H}_\mathcal{W}$ where $\mathcal{H}_\mathcal{O}$ is the output register, $\mathcal{H}_\mathcal{Q}$ the query register and $\mathcal{H}_\mathcal{W}$ a workspace register. Without loss of generality, we can make the initial and target states orthogonal by adding an ancilla qubit in state $\ket{0}$ for $\ket{\rho_x}$ and $\ket{1}$ for $\ket{\sigma_x}$. We then define a continuous path from $\ket{\rho_x}\ket{0}$ to $\ket{\sigma_x}\ket{1}$:
\[
\begin{split}
\ket{k_x^+(s)}_\mathcal{O}= &\quad \cos \theta(s)\ket{ 0,\rho_x}_\mathcal{O}  + \sin \theta(s)\ket{ 1,\sigma_x}_\mathcal{O}  , \\
\ket{k_x^-(s)}_\mathcal{O}= &      -\sin \theta(s)\ket{ 0,\rho_x}_\mathcal{O} + \cos \theta(s)\ket{ 1,\sigma_x}_\mathcal{O}   ,
\end{split} 
\]
where $\theta(s)= \frac{\pi}{2}s$ and  $s\in [0,1].$
% \vspace{0.3cm}\\

From Definition  \ref{def:adv}, let $\big\{\ket{u_{x,i}},\ket{v_{x,i}}\big\}_{x,i}$ be vectors witnessing $\gamma_2(\rho-\sigma\vert \Delta)=W$, with $W\overset{\mathrm{def}}{=}\advstar(\rho,\sigma)$.
We use those states to define the following non-normalized states:
\[
\begin{split}
\ket{\Psi_x^+(s,\varepsilon)} = & \ket{k_x^+(s)}_\mathcal{O} 
+\frac{\eps}{\sqrt{W}} \sum_i \ket{i, x^+_i}_\mathcal{Q}\ket{u_{x,i}}_\mathcal{W},\\
\ket{\Psi_x^-(s,\varepsilon)} = &\ket{k_x^-(s)}_\mathcal{O} 
+\xi(s)\frac{\sqrt{W}}{\varepsilon}\sum_i \ket{i, x^-_i}_\mathcal{Q}\ket{v_{x,i}}_\mathcal{W},
\end{split} 
\]
where $\ket{x_i^\pm}$ is defined by~\eqref{defPM}, and $\xi(s)=2 \cos\theta(s)\sin\theta(s)$.
Note that we have $\braket{x^-_i}{y^+_i}=\frac{1}{2}\big[1-\delta_{x_i,y_i}\big]$. We also let $\ket{\psi_x^\pm(s,\varepsilon)}$ be their normalized versions.

The algorithm uses as driver Hamiltonian the projection $\Lambda(s,\varepsilon)$ on the vector space $V(s,\varepsilon)=\mathrm{span}\{\ket{\Psi_x^-(s,\varepsilon)}\vert x\in \X\}$, and as oracle Hamiltonian, $\Pi_x=\sum_i\ketbra{i,x^-_i}{i,x^-_i}_\mathcal{Q}\otimes\bbone_\mathcal{W}$ (note that $\|\Pi_x\|\leq 1$).

\vspace{5mm}

\fbox{\begin{minipage}{0.75\textwidth}

\textbf{AdiaConvert}($\rho,\sigma,\varepsilon$) 
\begin{description}
\item[1] Prepare the state $\ket{ 0,\rho_x}$.
\item[2] If $\advstar(\rho,\sigma)< \varepsilon/2$, do nothing.
\item[3] Otherwise  apply the Hamiltonian $H_x(s,\varepsilon)=\Lambda(s,\varepsilon)-\Pi_x$,\\ where $s=t/\tau$  and  $\tau=15 \frac{\advstar(\rho,\sigma)}{\varepsilon^2}$, from $t=0$ to $t=\tau$.
\end{description}
\end{minipage}
}\\ \vspace{0.5cm}

The action of the algorithm is simple, first, if  $\advstar(\rho,\sigma)<\varepsilon/2$, then we claim, using Proposition \ref{lem:Dist} and  Corollary \ref{cor:relation}, that  $\rho$ and $\sigma$ are closed enough, and satisfies the coherent output condition given in Definition \ref{def:output}.

Otherwise, in order to convert the initial state $\ket{ 0,\rho_x}$ into a state close enough to the target state $\ket{ 1,\sigma_x}$, we consider the state $\ket{\psi_x^+(s,\eps)}$, which is $\eps$-distant to the state $\ket{k_x^+(s)}$ interpolating between the initial and target state. We then use the adiabatic process $\{H_x(s,\eps),P_x(s,\eps),\tau\}$ with failure $\eps$, where $P_x(s,\eps)$ is the rank-$1$ orthogonal projection on the state $\ket{\psi_x^+(s,\eps)}$. The correctness of the adiabatic evolution is  based on Lemma \ref{adiBound}, where the solution of Equation \eqref{eq:comm} follows from Item 5 in Proposition \ref{propAll}. Therefore the final state is $3\eps$-distant from the target state since the algorithm incurs error $\eps$ at the initial state, during the adiabatic process, and at the target state. This implies that we solve the quantum state generation problem with error at most $9\eps^2$, and in turn that $Q^\ct_{9\eps^2}(\rho,\sigma)\leq 15{\advstar(\rho,\sigma)}/{\varepsilon^2}$.

The proof of Theorem \ref{thm:adv} is the consequence of the existence of the optimal quantum query algorithm \textbf{AdiaConvert}. As the number of query involved are given by the time scale $\tau$, the demonstration relies on the derivation of an adiabatic bound linear in $\advstar$.\\

In order to prove Theorem \ref{thm:adv}, we first derive several useful properties of the algorithm \textbf{AdiaConvert}.
\begin{prop}\label{propAll}
For any $s,\eps\in[0,1]$ and for all $x\in \X $
\begin{description}
\item[1] $N_x(\eps)\overset{\mathrm{def}}{=}\|\ket{\Psi_x^+(s,\eps)}\|\leq 1+\eps^2/2$,
\item[2] $\ket{k_x^+(s)}$ and $ \ket{\psi_x^+(s,\eps)}$ are $\eps$-distant,
\item[3] $\Lambda(s,\eps)\ket{\psi_x^+(s,\eps)}=0, $
\item[4] $\ket{\psi_x^+(s,\eps)}$ is an eigenvector of $H_x(s,\eps)$ with eigenvalue $\lambda_x(s,\eps)=0,$ 
\item[5] $\bra{\psi_x^+(s,\eps)}\Big(\partial_s \ket{\psi_x^+(s,\eps)}\Big)=0,$
\item[6] $\partial_s\ket{\Psi_x^+(s,\varepsilon)}= \frac{\pi}{2} H_x(s,\varepsilon)\ket{\Psi_x^-(s,\varepsilon)},$
\item[7] $\|\ket{\Psi_x^-(s,\eps)}\|^2\leq 1+W^2/\eps^2.$
\end{description}
\end{prop}

Let us note that Item \textbf{5} is the key property that prevents the instantaneous state $\ket{\psi_x^+(s,\eps)}$ from leaking to degenerate subspaces of eigenvalue 0.

\begin{proof}
\textbf{1)} By Definition \ref{def:filter}, we have
$
 \sum_{i}\|\ket{u_{x,i}}\|^2\leq\gamma_2(\rho-\sigma|\Delta)=W
$,
so that
\begin{align*}
 N_x^2(\eps) =\Big\|\ket{\Psi_x^+(s,\varepsilon)} \Big\|^2 = 1 + \frac{\varepsilon^2}{ W}\sum_{i}\Big\|\ket{u_{x,i}}\Big\|^2\leq 1+ \varepsilon^2.
\end{align*}
Item \textbf{1} then follows from the inequality $\sqrt{1+\delta}\leq 1+\delta/2$, for $\delta\in[0,1]$.
\vspace{0.2cm}\\

\textbf{2)} The scalar product of these vectors gives 
\begin{align*}
 \braket{\psi_x^+(s,\varepsilon)}{k^+_x(s)}=\frac{1}{N_x(\eps)}\braket{\Psi_x^+(s,\varepsilon)}{k^+_x(s)}=\frac{1}{N_x(\eps)}\geq 1-\varepsilon^2 /2.
\end{align*}
Since this scalar product is real, we have
\begin{align*}
 \norm{\ket{k_x^+(s)}-\ket{\psi_x^+(s,\varepsilon)}}^2=2-2\braket{\psi_x^+(s,\varepsilon)}{k^+_x(s)}\leq\eps^2.
\end{align*} 
\vspace{0.2cm}

\textbf{3)} Remember $\Lambda(s,\eps)$ is the projection on subspace $V(s,\varepsilon)=\mathrm{span}\{\ket{\Psi_x^-(s,\varepsilon)}\vert x\in \X\}$. Therefore, it suffices to show that for all $x,y\in\X$,
$\braket{\Psi_x^+(s,\varepsilon)}{\Psi_y^-(s,\varepsilon)}=0.$
By definition of $\ket{\Psi_x^+(s,\varepsilon)}$ and $\ket{\Psi_x^-(s,\varepsilon)}$, we have
\begin{align*}
 \braket{\Psi_x^+(s,\varepsilon)}{\Psi_y^-(s,\varepsilon)} =-\cos \theta(s)\sin\theta(s)\big[\rho_{x,y}-\sigma_{x,y} - \sum_{j: x_j\neq y_j}\braket{u_{x,j}}{v_{y,j}} \big].
\end{align*}
The right hand side is then zero due to the properties of $\big\{\ket{u_{x,i}},\ket{v_{x,i}}\big\}_{x,i}$ in Definition~\ref{def:adv}.
\vspace{0.2cm}\\

\textbf{4)} From Item \textbf{3} we already know that $\Lambda(s,\eps)\ket{\psi_x^+(s,\eps)}=0$. Then by the definition of $H_x(s,\eps)$, we must calculate $\Pi_x\ket{\psi_x^+(s,\eps)}$,
\begin{align*}
 \Pi_y\ket{\psi_x^+(s,\varepsilon)}\propto \sum_i [1-\delta_{x_i,y_i}]\ket{i,x^+_i,u_{x,i}},
\end{align*}
which is exactly zero for $x=y$.  
\vspace{0.2cm}\\

\textbf{5)} The property follows from
\begin{align*}
 \partial_s\ket{\psi_x^+(s,\eps)}=\frac{1}{N_x(\eps)}\partial_s\ket{\Psi_x^+(s,\eps)}=\frac{\pi }{2N_x(\eps)}\ket{k_x^-(s)}
\end{align*}
and the fact that $ \braket{\psi_x^+(s,\eps)}{k_x^-(s)}\propto\braket{k_x^+(s)}{k_x^-(s)}=0$.
\vspace{0.2cm}\\

\textbf{6)} 
\[
\begin{split}
\partial_s\ket{\Psi_x^+(s,\varepsilon)}&=\frac{\pi}{2} \ket{k_x^-(s)}\\
&=\frac{\pi}{2}  \Big(\bbone-\Pi_x\Big)\ket{ \Psi_x^-(s,\varepsilon)}\\
&=\frac{\pi}{2} 
\Big[ \Big(\Lambda(s,\varepsilon)-\Pi_x\Big)+\Big( \bbone-\Lambda(s,\varepsilon)\Big) \Big]\ket{  \Psi_x^-(s,\varepsilon)}\\
&=\frac{\pi}{2} H_x(s,\eps) \ket{\Psi_x^-(s,\varepsilon)}.
\end{split}
\]
In the second line,   $\Pi_x$ acts as the identity on $\ket{i,x_i^-}$. In the third line, the second term is zero by definition of $\Lambda(s,\varepsilon)$.
\vspace{0.2cm}\\

\textbf{7)} 
Similarly to the proof of Item \textbf{1} all vectors $\ket{v_{x,i}}$ have their norm bounded by $W$
\begin{align*}
 \Big\|\ket{\Psi_x^-(s,\varepsilon)} \Big\|^2 = 1 + \xi^2(s)\frac{W}{ \varepsilon^2 }\sum_{i}\Big\|\ket{v_{x,i}}\Big\|^2\leq 1+ \frac{W^2}{ \varepsilon^2 }.
\end{align*}
Noting that $\xi(s)=\sin(2\theta(s))$.

\end{proof}

% We are now ready to prove Theorem~\ref{thm:adv}.

\begin{proof}[Proof of Theorem~\ref{thm:adv}]  ~\\

Let $W=\advstar(\rho,\sigma)$.
We show that \textbf{AdiaConvert} solves the quantum state conversion in time $\tau=15 \frac{W}{\varepsilon^2}$ with error at most $9\eps^2$. Let us first consider the case where $W<\eps/2$. Then, Proposition~\ref{lem:Dist} implies $\mathcal{D}_H(\rho,\sigma)< \varepsilon/2$, and Corollary~\ref{cor:relation} concludes that $\mathcal{F}_H(\rho,\sigma)> 1-\eps/2 > \sqrt{1-\eps}$, so that the coherent output condition is already satisfied by the initial Gram matrix.

We now assume that $W\geq \eps/2$. Before we go any further, we must justify that the triplet $\{H_x(s,\eps),$ $P_x(s,\eps),\tau\}$ is an adiabatic process as defined in Definition \ref{def:adia}. First by definition, the state $\ket{\psi_x^\pm(s,\eps)}$ is $s$-smooth on $[0,1]$. It follows that $H_x(s,\eps)$  and $P_x(s,\eps)$ are also $s$-smooth. Moreover, by Item \textbf{4} of Proposition \ref{propAll}, $\ket{\psi_x^+(s,\eps)}$ is an eigenstate of $H_x(s,\eps)$ with a constant eigenvalue $\lambda_x(s,\eps)=0$.

In order to bound the error of the adiabatic process $\eps_{AP}$ with  Lemma \ref{adiBound}, we define an operator $X_x(s,\eps)$, solution of Equation \eqref{eq:comm}, where $X_x(s,\eps)$ and $\dot{X}_x(s,\eps)P_x(s,\eps)$ are bounded. 

Let $X_x(s,\eps)=\frac{\pi}{2N_x(\eps)}\ket{\Psi^-_x(s,\eps)}\!\!\bra{\psi^+_x(s,\eps)}$, Items \textbf{5} and \textbf{6} of Proposition \ref{propAll} imply 
\[[H_x(s,\eps),X_x(s,\eps)]=H_x(s,\eps)X_x(s,\eps)=\dot{P}_x(s,\eps)P_x(s,\eps).\]
To obtain $\eps_{AP}$ we derive a bound for $X_x(s,\eps)$ and $\dot{X}_x(s,\eps)P_x(s,\eps)$.

First, we have
\begin{align*}
 \|X_x(s,\eps)\|^2=\Big[\frac{\pi}{2N_x(\eps)}\Big]^2\Big\|\ket{\Psi_x^-(s,\varepsilon)} \Big\|^2.
\end{align*}
From Item \textbf{7} of Proposition \ref{propAll} and the fact that $W\geq \eps/2$, we obtain 
\begin{align*}
\|\ket{\Psi_x^-(s,\eps)}\|^2\leq 1+ \frac{W^2}{\eps^2}\leq 5 \frac{W^2}{\eps^2},
\end{align*}
knowing that $N_x(\eps)\geq 1$ we obtain the bound :  $\|X_x(s,\eps)\|\leq \frac{\pi \sqrt{5} }{2 }\frac{W}{\eps} $.\\

Second,   to bound $\|\dot{X}_x(s,\eps)P_x(s,\eps)\|$ we  derive $X_x(s,\eps)$
\[\dot{X}_x(s,\eps)=\frac{\pi}{2N_x(\eps)}\partial_s\big(\ket{\Psi^-_x(s,\eps)}\!\big)\!\bra{\psi^+_x(s,\eps)}+\frac{\pi^2}{4N_x(\eps)}\ket{\Psi^-_x(s,\eps)}\!\!\bra{k^-_x(s)}.\]
After adding $P_x(s,\eps)$ on the right, the second term disappears following Item \textbf{5} of Proposition \ref{propAll}, and we have
\[\begin{split}
\|\dot{X}_x(s,\eps)P_x(s,\eps)\|^2&=\Big[\frac{\pi}{2N_x(\eps)}\Big]^2\Big\|\partial_s\ket{\Psi^-_x(s,\eps)}\Big\|^2\\
&\leq \Big[\frac{\pi}{2 }\Big]^2 \Big(   \frac{\pi^2}{4} + \pi^2 \cos^2(\pi s) \frac{ W}{\varepsilon^2}\sum_{i}\|\ket{v_{x,i}}\|^2 \Big)\\
&\leq \Big[\frac{\pi}{2 }\Big]^2\pi^2\Big( \frac{1}{4} +\frac{W^2}{\varepsilon^2}\Big)\\
&\leq  \Big[\frac{\pi}{2 }\Big]^2 2\pi^2\frac{ W^2}{\eps^2}.
\end{split}\]

Thereby we have all the required conditions to use Lemma \ref{adiBound} for the adiabatic process  $\{H_x(s,\eps),P_x(s,\eps),\tau\}$, which ensures that $\eps_{AP}\leq\eps$ if
\[\tau\geq  \frac{15W}{\eps^2}\geq \frac{1}{\eps}\Big[\frac{W}{\eps}\Big( \pi\sqrt{5}+\frac{\pi^2}{\sqrt{2} }\Big) \Big].\]
Let $\ket{\psi_x^f}$ be the output state. Since the initial state  $\ket{ 0,\rho_x}$ and the target state  $\ket{ 1,\sigma_x}$ are $\eps$-distant from $\ket{\psi_x^+(0,\eps)}$ and $\ket{\psi_x^+(1,\eps)}$ (Item \textbf{2} of Proposition \ref{propAll}) and the adiabatic process introduces an additional error of $\eps_{AB}\leq\eps$, the output state $\ket{\psi_x^f}$ and the target state $\ket{ 1,\sigma_x}$ are $3\eps$-distant, which implies that $\Re(\braket{\psi_x^f}{1,\sigma_x})\geq \sqrt{1-9\eps^2}$. Therefore, we obtain
\begin{align*}
Q^\ct_{9\eps^2}(\rho,\sigma)\leq 15\frac{W}{\varepsilon^2},
\end{align*}
which implies the theorem by setting $\eps'=9\eps^2$.
\end{proof}

% \hfill$\blacksquare$

\subparagraph*{Acknowledgements}

This work was supported by the Belgian ARC project COPHYMA and the European Union Seventh Framework Programme (FP7/2007-2013) under grant agreement n. 600700 (QALGO).

\bibliographystyle{alphaurl}
\bibliography{adiabatic-quantum-query}

\appendix

\section{Appendix: Adiabatic theorem without a gap condition}\label{sect:bound}

In this section we give an adapted version of the proof of Lemma \ref{adiBound} in \cite{Avron1998}. We derive an upper bound on the error $\eps_{AP}$ caused by the adiabatic process without a gap condition.  We use the same notations as in Subsection \ref{sect:adia}. 
\begin{lem}\cite{Avron1998}
Let $\{H(s),P(s),\tau\}$ be an adiabatic process and $\eps>0$. Suppose that the commutator equation 
\begin{equation}\label{eq:commm}
\dot{P}(s)P(s)=[H(s),X(s)]
\end{equation}
accepts as solution operator $X(s)$ such that both $X(s)$ and $\dot{X}(s)$ are bounded. If
\begin{align*}
 \tau \geq \max_{s\in [0,1]}\frac{1}{\eps}\Big[2\|X(s)\|+\|\dot{X}(s)P(s)\|\Big] ,
\end{align*}
then $\eps_{AP}\leq\eps$.
\end{lem}

\paragraph*{Proof of Lemma~\ref{adiBound}}

In order to bound the quantity $\eps_{AP}$, we would like to describe an idealized adiabatic evolution $U_A(s)$ that transports the projector $P(0)$ to $P(s)$, such that $U_A(s)P(0)=P(s)U_A(s)$. To achieve this, we use a technique given by \cite{Kato1950} (later improved in~\cite{Avron1987}), and define $H_A(s)$ as the \textit{adiabatic Hamiltonian}
\begin{equation}\label{hamAdia} 
H_A(s)=\lambda(s)\bbone + \frac{i}{\tau}[\dot{P}(s),P(s)],
\end{equation}
where $[\cdot,\cdot]$ is the commutator.
%The condition~\ref{con:ana} extends to $P(s)$, $\dot{P}(s)$ and $H_A(s)$, ensuring
We define $U_A(s)$ as the solution of the Schr\"odinger equation for this Hamiltonian, that is,
\begin{equation}
i\partial_s U_A(s)=\tau H_A(s)U_A(s),
\end{equation}
with the initial condition $U_A(0)=\bbone$. The existence and uniqueness of $U_A(s)$ follows from the analytical properties in Definition \ref{def:adia}. Moreover we show that $U_A(s)$ has the desired property.

\begin{lem}\cite{Kato1950} (Intertwining property)\label{lem:intertwining}
\begin{equation}
U_A(s)P(0)=P(s)U_A(s).\label{twin}
\end{equation} 
\end{lem}
The proof of this property uses the following fact.
\begin{fact}\label{fone}For any orthogonal projector $P$ we have $P=P^2$, so that $\dot{P}=\dot{P}P+P\dot{P}$ and $P\dot{P}P=0$ .\end{fact}

\begin{proof}[Proof of Lemma~\ref{lem:intertwining}]
Since $U_A(s)$ is the solution of the differential equation $i\partial_sX(s)=\tau H_A(s)X(s)$ with $X(0)=\bbone$, then every other solution of this equation has the form $X(s)=U_A (s) X(0)$. All we need to do is prove that $P(s)U_A(s)$ is also a solution. Indeed, this implies that $P(s)U_A(s)=U_A (s) X(0)$, and by setting $s=0$ we obtain $P(0)=X(0)$. Using Fact~\ref{fone}, we have
\[\begin{split}
i\partial_s\big(P(s)U_A(s)\big)&=i\dot{P}(s)U_A(s)+P(s)\tau H_A(s)U_A(s)\\
&=i\dot{P}(s)U_A(s)+\tau \lambda(s)P(s) U_A(s)+iP(s)[\dot{P}(s),P(s)]U_A(s)\\
&=\tau \lambda(s)P(s) U_A(s)+i\big(\dot{P}(s)-P(s)\dot{P}(s)\big)U_A(s)\\
&=\tau\lambda(s) P(s) U_A(s)+i\dot{P}(s) P(s)U_A(s)\\
&=\big(\tau\lambda(s) \bbone +i[\dot{P}(s),P(s)]\big)P(s)U_A(s)\\
&=\tau H_A(s)P(s)U_A(s)\\
\end{split}
\]
\end{proof}

In order to prove Lemma~\ref{adiBound}, we need two more claims. \\Note that $\eps_{AP}(s)$ can be rewritten as $\|\big(\Omega(s)-\bbone\big) P(0)\|$, where  $\Omega(s)=U^*_\tau(s)U_A(s)$. 
\begin{claim}\label{ftwo}
$\dot{\Omega}(s)P(0)=U_\tau^*(s)\dot{P}(s)U_A(s)P(0)$
\end{claim}

\begin{proof}Using~\eqref{tau} and~\eqref{hamAdia}, we note that $\dot{\Omega}(s)=U_\tau^*(s)\Big[i\tau\big(H(s)-\lambda(s)\bbone\big)+[\dot{P}(s),P(s)]\Big]U_A(s)$.
The claim follows from the intertwining property (Lemma~\ref{lem:intertwining}), Fact~\ref{fone} and $H(s)P(s)=\lambda(s)P(s)$.
\end{proof}

\begin{claim}\label{fthree}
Let $\Phi(s)=e^{-i\tau\lambda(s)}\bbone$ and $V_A(s)=\Phi^*(s)U_A(s)$. Then $V_A(s)$ satisfies the intertwining property~\eqref{twin}, that is, $V_A(s)P(0)=P(s)V_A(s)$, as well as the Schr\"odinger equation
$
 \dot{V}_A(s)=[\dot{P}(s),P(s)]V_A(s).
$
\end{claim}
\begin{proof}
The fact that $V_A(s)$ satisfies the intertwining property is immediate since $U_A(s)$ satisfies this property and  $\Phi(s)$, being proportional to the identity, commutes with any operator.
The fact that it satisfies the Schr\"odinger equation follows from the facts that $\Phi(s)$ satisfies $i\dot{\Phi}(s)=\tau\lambda(s)\Phi(s)$, $U_A(s)$ satisfies $i\dot{U}_A(s)=\tau H_A(s)U_A(s)$, and both terms of $H_A(s)=\lambda(s)\bbone + \frac{i}{\tau}[\dot{P}(s),P(s)]$ commute.
\end{proof}

Let $X(s)$ an operator solution of $\dot{P}(s)P(s)=[H(s),X(s)],$ then
\[
\begin{split}
\big(\Omega(s)-\bbone \big)P(0)&=\int_0^s\dot{\Omega}(s')ds'P(0)\\
&=\int_0^s U_\tau^*(s')\dot{P}(s')U_A(s')ds'P(0)\\
&=\int_0^s U_\tau^*(s')\Phi(s')\dot{P}(s')V_A(s')ds'P(0)\\
&=\int_0^s U_\tau^*(s')\Phi(s')[H(s'),X(s')]V_A(s')ds'P(0)\\
&=\int_0^s U_\tau^*(s')\Phi(s')[H(s')-\lambda(s')]X(s')V_A(s')ds'P(0)\\
&=\frac{1}{i\tau}\int_0^s \partial_{s'}[{U}_\tau^*(s')\Phi(s')]X(s')V_A(s')ds'P(0)\\
&=\frac{1}{i\tau}\Big[ U_\tau^*(s')\Phi(s')X(s')V_A(s')\Big]^s_0 P(0)
-\frac{1}{i\tau}\int_0^s U_\tau^*(s')\Phi(s')\partial_{s'}[X(s')V_A(s')]ds'P(0)\\
&=\frac{1}{i\tau}\Big[ U_\tau^*(s')X(s')U_A(s')\Big]^s_0 P(0)-\frac{1}{i\tau}\!\int_0^s U_\tau^*(s')[\dot{X}(s')+X(s')\dot{P}(s')]U_A(s')ds'P(0)
\end{split}
\]
We explain line by line:
\begin{enumerate}
\item[$(1\rightarrow 2)$] We use Claim~\ref{ftwo}.
\item[$(2\rightarrow 3)$] We rearrange the expression using $U_A(s)=\Phi(s)V_A(s)$ and the fact that $\Phi(s)$ commutes with any operator.
\item[$(3\rightarrow 4)$] We use the intertwining property for $V_A(s)$ (Claim~\ref{fthree}) and Equation \eqref{eq:commm}.
\item[$(6\rightarrow 7)$] We integrate by parts.
\end{enumerate}
The third term in the last line is null, because $X(s)=X(s)P(s)$ and the intertwining property (Lemma~\ref{lem:intertwining}) yields the expression $P\dot{P}P$, which is zero by Fact~\ref{fone}. Using the triangle inequality, the fact that a norm is preserved by unitary operations and can only decrease under projections, we finally have
\[\begin{split}
\eps_{AP}(s)&=\|\big(\Omega(s)-\bbone\big)P(0)\|\\
&\leq \frac{1}{\tau}\Big[ \|X(0)\|+\|X(s)\|+s\max_{s'\in[0,s]}\|\dot{X}(s')P(s')\|\Big]\\
&\leq \frac{1}{\tau} \max_{s\in[0,1]}\Big[ 2\|X(s)\|+\|\dot{X}(s)P(s)\|\Big]
\end{split}
\]

This conclude the proof.\hfill$\blacksquare$

\end{document}